\documentclass[letterpaper, 10pt, conference,twocolumn]{ieeeconf} 

\IEEEoverridecommandlockouts    
\usepackage{psfrag}
\usepackage{amsmath}
\usepackage{amssymb}

\usepackage{amsfonts}
\usepackage{latexsym}
\usepackage{graphicx}

\usepackage{colortbl}
\usepackage{fancyhdr}
\usepackage{epstopdf}
\usepackage{float,bbm}
\usepackage{subfigure}
\usepackage{subfig}

\usepackage{enumerate}
\usepackage{multicol}
\usepackage{graphicx}

\usepackage[usenames,dvipsnames,svgnames,table]{xcolor}

\newtheorem{theorem}{Theorem}
\newtheorem{definition}{Definition}
\newtheorem{assumption}{Assumption}
\newtheorem{proposition}{Proposition}
\newtheorem{lemma}{Lemma}

\newtheorem{remark}{Remark}
\newcommand {\be}{\begin{equation}}
\newcommand {\ee}{\end{equation}}

\newcommand {\bes}{\begin{equation*}}
\newcommand {\ees}{\end{equation*}}

\newcommand{\N}{\mathbb{N}}
\newcommand{\R}{\mathbb{R}}
\newcommand{\mc}[1]{\mathcal{#1}}
\newcommand{\q}{z}

\newenvironment{ps}
  {\left[\begin{smallmatrix}}
  {\end{smallmatrix}\right]}
  
\usepackage{multirow}
\usepackage{rotating}

\usepackage{soul}

\title{\LARGE \bf
Sensitivity analysis for network aggregative games}

\IEEEoverridecommandlockouts
\author{Francesca Parise, Asuman Ozdaglar
\thanks{F. Parise and A. Ozdaglar are with the Laboratory for Information and Decision Systems, Massachusetts Institute of Technology, Cambridge, MA:  {\tt\small \{parisef,asuman\}@mit.edu}. Research supported by the SNSF grant number P2EZP2 168812 and partially supported by ARO MURI W911NF-12-1-0509.}
         }

\begin{document}

\maketitle
\begin{abstract}
 We investigate the sensitivity of the Nash equilibrium of constrained network aggregative  games to changes in exogenous parameters affecting the cost function of the players. 
This setting is motivated by two applications. The first is the analysis of interventions by a social planner with a networked objective function while the second is network routing games with atomic players and information constraints. By exploiting a primal reformulation of a sensitivity analysis result for variational inequalities, we provide a characterization of the sensitivity of the Nash equilibrium that depends on  primal variables only.  To derive this  result we assume strong monotonicity of the mapping associated with the game. As the second main result, we derive sufficient conditions that guarantee this strong monotonicity property in network aggregative games. These two characterizations allows us to systematically study changes in the Nash equilibrium due to perturbations or parameter variations in the two applications mentioned above.
\end{abstract}

\section{Introduction}

Network aggregative games (NAGs) are games in which the cost function of each agent depends on its own strategy and on an aggregate of the strategies of its neighbours, as defined by an underlying interaction network. NAGs are used  to model a vast range of applications  spanning from  sociology \cite{ghaderi2014opinion}  and economics \cite{acemoglu2015networks} to traffic \cite{roughgarden2007routing} and energy markets \cite{ma:callaway:hiskens:13,chen2014autonomous}. 
In many cases, the resulting models depend on a vector of parameters that might represent either some exogenous factor in the case of social applications (e.g.  rumors in opinion dynamics or  shocks in financial networks) 
or a tuning variable in the case of technological applications (e.g. road improvements in traffic networks  or the price elasticity in energy markets). In both cases, to fully understand (and possibly control) the behaviour of the system it is important to study the effect that variations of such parameters  have on the final outcome of the game. 
As a first step in this direction,  we here focus on  the sensitivity of the Nash equilibrium to small changes in the cost functions.  

 Our first main contribution is the derivation of a general \textit{sensitivity result for the Nash equilibrium of games with constraints}. To this end, we use an implicit function theorem type of argument to get a sensitivity result for the KKT system of the variational inequality associated with the game, as done in \cite{friesz2016foundations}. Our key idea is then to do a primal reformulation to get a sensitivity result that involves only the decision variables (i.e., the players strategies) and not the dual variables. This reformulation in terms of primal variables only  is new to the best of our knowledge and allows us to derive useful insights for two important classes of NAGs discussed next.
    The only other formulation in terms of primal variables  that we are aware of is derived in \cite{dafermos1988sensitivity}. Therein, however, a  more general setup is assumed and, as a consequence, the obtained formula  conceals the effect of the constraints on the Nash equilibrium. This is, on the other hand, immediate from  our reformulation.
 
 To derive our first result we need to assume strong monotonicity of the operator of the variational inequality associated with the game. 
  As the second main contribution, we derive   \textit{two sufficient conditions to guarantee strong monotonicity} of this operator in the case of  NAGs. The formulation of NAGs that we use here is the one  presented  in \cite{parise2017network}. We note that sufficient conditions for the case of average aggregative games, that is, games where the cost function of each agent depends on the average of the whole population (instead of the average of the neighbours) have been derived  in \cite{paccagnan2016aggregative,gentile2017nash} in terms of the cost function. The  conditions for NAGs that we derive here instead depend both on the cost function and on the network.

 The first specific class of NAGs  that we consider is that of \textit{quadratic network games} \cite{jackson2014games}. Here we assume that for each agent $i$ there is a  parameter $y^i$ that represents a small  stochastic perturbation or shock  to the cost function of agent~$i$.  Without interventions, these shocks would drive  the system from the initial Nash equilibrium $x^\star(0)$ (in the absence of shocks)  to the new Nash equilibrium configuration  $x^\star(y)$. We first show that the  sensitivity results in \cite{acemoglu2015networks}  can be found as a corollary of our main result. We then consider a controlled setting in which a central authority  can intervene on the system and guarantee that the strategy of one agent $k$ in the new equilibrium  is the same as in the case without shocks. In other words, we assume that the central authority can impose the constraint  $x^\star_k(y)=x^\star_k(0)$. We use our main result to characterize which agent should be selected by the central authority to minimize the effect of the shocks on the whole system. As a corollary of this analysis, if the  adjacency matrix of the network is non-negative, one can recover the  key-player result derived in \cite{ballester2006s} for linear-quadratic network games. We demonstrate the effectiveness of our approach on a model of rumor spreading in opinion dynamics. Other applications of  quadratic network games where our results could apply are  production networks and systemic risk in financial networks as described  in \cite{acemoglu2015networks} or  crime, education and urban dynamics as described in \cite{jackson2014games,ballester2006s}.

  The second  specific class of NAGs that we consider  is that of \textit{atomic routing games} \cite{roughgarden2007routing}.  Specifically, we consider games where each player is a user that needs to allocate a fixed amount of splittable flow from a given origin to a given destination on a traffic network. The cost function of each user is its travel time, which depends on the strategy of the other users because of congestion effects. Here we assume that there is one parameter $y^e$ for each road $e$ of the traffic network. Variations of the parameter  $y^e$ are determined by a central authority and model the effort made by the authority to reduce the travel time on road $e$. Our objective is to characterise which road should the central authority  improve  to optimize the  system performance (i.e. to reduce the total travel time for all the users).   
 Besides the natural network flow constraints, we additionally use   constraints  to model the fact that the  users might have different sets of  information on the network, as motivated in \cite{acemoglu2016informational}. For example different agents may know different subsets of roads. Finally, we note that  the standard sensitivity results derived in \cite{qiu1989sensitivity, josefsson2007sensitivity, lu2008sensitivity, do2014sensitivity} do not apply in our setting because, on the one hand, we consider agents with information constraints and, on the other hand, we consider atomic agents (i.e., we focus on the Nash equilibrium instead of the Wardrop equilibrium).

Besides the already mentioned works, our results have strong connections with sensitivity results of  Nash equilibria  that have been derived in the economics literature (referred to as comparative statics).  Most of the economics literature has  focused on games with specific ordering properties (e.g., games with strategic complements/substitutes  as discussed in \cite{dubey2006strategic,milgrom1994monotone}) or games where the aggregate of the strategies is a scalar quantity (e.g.,   scalar aggregative games have been treated in \cite{jensen2010aggregative,acemoglu2013aggregate,bramoulle2014strategic}).
The main difference of our work from  this line of  results is that we consider explicitly the effect of the constraints.

The rest of the paper is organized as follows. In Section~\ref{NAG} we introduce network aggregative games and our main result on sensitivity of the Nash equilibrium in constrained games. In Section \ref{sec:SMON} we derive two sufficient conditions for strong monotonicity in NAGs. In Section \ref{systemic_risk} and \ref{traffic_sec} we show how these results can be applied to quadratic network games and atomic routing games, respectively. 

\subsubsection*{Notation} We denote the Jacobian of  a function $f(x):\R^n\rightarrow \R^d$  by $\nabla_x f(x) \in \R^{d \times n}.$ Given $a,b\in\N$, $\N[a,b]$ denotes the set of integer numbers in the interval $[a,b]$.
 $I_n$ denotes the $n$-dimensional identity matrix, $\mathbbm{1}_n$ the vector of unit entries and $e_i$ the $i$th canonical vector.
Given $A\in\mathbb{R}^{n\times n}$ not necessarily symmetric, $A\succ0$ ($\succeq0$) $\Leftrightarrow$ $x^\top A x>0~(\ge0),$ $\forall x\neq 0$,    $A_{(k,:)}$ denotes the $k$th row of $A$,  $\rho(A)$ denotes the spectral radius of $A$ and $\Lambda(A)$ its spectrum. $A\otimes B$ denotes the Kronecker product and $[A;B]:=[A^\top, B^\top ]^\top$. Given $N$ matrices $\{A^i\}_{i=1}^N$, $\mbox{blkdiag}(A^1,\ldots,A^N)$ is the block diagonal matrix whose $i$th block is  $A^i$. Given $N$ vectors $x^i\in\mathbb{R}^{n}$, $[x^1;\ldots;x^N]:=[x^i]_{i=1}^N:=[{x^1}^\top,\ldots ,{x^N}^\top]^\top\in\R^{Nn}$. 
 $|\mathcal{E}|$ denotes the cardinality of the set $\mathcal{E}$. 

\section{ Network aggregative games}
\label{NAG}

\subsection{ The model}

We consider a game with $N$ players that are interacting over a network with non-negative  adjacency matrix $P$, \cite{parise2017network}. We assume $P_{ii}=0$ for all $i\in\N[1,N]$. We say that $j$ is a neighboor of $i$ if $P_{ij}>0$. The set of neighbours of player $i$ is denoted by $\mathcal{N}^i$.
Each player $i\in\N[1,N]$ aims at selecting a vector strategy $x^i\in\R^n$ among its feasible set $\mathcal{X}^i\subseteq \R^n$  to \textit{minimize} a  cost function
\begin{equation}\label{gamey}
J^i(x^i,\q^i(x),y^i),
\end{equation}
which depends on its own strategy $x^i$, on the linear combination of the neighbours strategies  according to the coefficients of the network  $$\textstyle \q^i(x):=\sum_{j=1}^n P_{ij}x^j=\sum_{j\in\mathcal{N}^i} P_{ij}x^j$$ and on a parameter $y^i\in\mathcal{Y}^i\subseteq \R^{d_i}$. Let $y:=[y^i]_{i=1}^N$ and $\mathcal{Y}=\mathcal{Y}^1\times \ldots \times \mathcal{Y}^N$.
 For each admissible parameter vector $y\in\mathcal{Y}$ we refer to the game with $N$ players and cost functions as in~\eqref{gamey} as the network aggregative game (NAG) with parameter $y$ and we denote it  by $\mathcal{G}(y)$. We assume that for each $y\in\mathcal{Y}$,  the game $\mathcal{G}(y)$ satisfies the following regularity conditions. 
\begin{assumption}[Constraints]\label{constraints}
The  constraint sets can be expressed as $\mathcal{X}^i:=\{x^i\in\R^n\mid B^i x^i\le b^i, H^i x^i=h^i\}$, for some $B^i\in\R^{m_i \times n},  b^i\in\R^{m_i}, H^i\in\R^{p_i \times n},  h^i\in\R^{p_i}$. 
There exists $x^i\in\mathcal{X}^i$ such that $B^ix^i<b^i$. 
\end{assumption}

In the following, we define  $x:=[x^i]_{i=1}^N\in\R^{Nn}$ the vector whose $i$-th block component is the  strategy of agent~$i$. We also define $m:=\sum_{i=1}^N m_i, p:=\sum_{i=1}^N p_i$, $B:=\mbox{blkdiag}(B^1,\ldots, B^N) \in\R^{m\times Nn}, b=[b^i]_{i=1}^N \in\R^{m}, H:=\mbox{blkdiag}(H^1,\ldots, H^N) \in\R^{p \times Nn}, h=[h^i]_{i=1}^N \in\R^{p}$, so that $\mathcal{X}:=\mathcal{X}^1\times\ldots\times \mathcal{X}^N=\{x\in\R^{Nn}\mid Bx\le b, Hx=h\}$.
\begin{assumption}[Cost function]\label{cost}
For all $y^i\in\mathcal{Y}^i$,  the function $x^i\mapsto J^i(x^i,\q^i(x),y^i)$ is twice continuously differentiable and  strongly convex in $x^i$  for all $i\in\N[1,N]$  and for all $x^{-i}\in\mathcal{X}^{-i}$. Moreover, $[x^i;z^i]\mapsto J^i(x^i,z^i,y^i)$ is twice continuously differentiable in $[x^i;z^i]$.
\end{assumption}

In the rest of the paper we are going to study properties of the Nash equilibrium of the game $\mathcal{G}(y)$, as defined next. 
\begin{definition}[Nash equilibrium]\label{nash}
A set of strategies $\{ x^{\star i}(y) \in \mathcal{X}^i\}_{i=1}^N$ is a Nash equilibrium of the game $\mathcal{G}(y)$ if for all players $i\in\N[1,N]$ we have that for all $ x^i\in\mathcal{X}^i$
\begin{align*}
\textstyle J^i(x^{\star i}(y),\sum_{j\in\mathcal{N}_i} P_{ij} &x^{\star j}(y) ,y^i)\\ &\textstyle \le  J^i(x^i,\sum_{j\in\mathcal{N}_i} P_{ij}x^{\star j}(y) ,y^i).
\end{align*}
\end{definition}

\subsection{ Connection between game theory and  variational inequalities}

We start by recalling an equivalent characterization of the Nash equilibrium of a generic game (i.e. not necessarily a NAG) in terms of variational inequalities. 
\begin{definition}[Variational inequality]\label{vi}
A vector $\bar x\in\R^n$ solves the variational inequality VI$(\mathcal{X},f)$ with set $\mathcal{X}\subseteq\R^n$ and operator 
$f(x):\mathcal{X} \rightarrow \R^n$ if and only if
$$f(\bar x)^\top (x-\bar x) \ge 0, \quad \forall x\in\mathcal{X}.$$
\end{definition}
To this end, we define the operator $F(x, y):\mathcal{X}\times\mathcal{Y}\rightarrow \R^{Nn}$ whose $i$-th block component is the gradient of the cost function of agent $i$ with respect to its own strategy for the game $\mathcal{G}(y)$. Mathematically,
\begin{align}\label{eq:F}
F(x, y)&:=[\nabla_{x^i} J^i(x^i,\q^i(x),y^i)^\top]_{i=1}^N.
\end{align}
This operator   is  of fundamental importance in the subsequent analysis because of the following well known relation.
\begin{proposition}[VI and KKT system]\label{prop:kkt}
Under Assumptions \ref{constraints} and \ref{cost}  the following statements are equivalent
\begin{enumerate}
\item a vector of strategies $x^\star(y)$ is a Nash equilibrium for the game  $\mathcal{G}(y)$; 
\item $x^\star(y)$ solves the VI$(\mathcal{X},F(\cdot,y))$;
\item  there exists $\lambda(y)\in \R^{m} $ and $\mu(y)\in \R^{p} $
such that
\begin{subequations}\label{KKT}
\begin{align}
&F(x^\star(y),y) + B^\top\lambda(y) + H^\top\mu(y)=0 \label{KKT1} \\
& H x^\star(y)=h, \quad \lambda(y)^\top (B x^\star(y)-b)=0, \label{KKT2}\\
& B x^\star(y)\le b,  \quad \lambda(y) \ge 0.
\end{align}\end{subequations}
\end{enumerate}
\end{proposition}
\begin{proof} This is a classic result, see e.g. \cite[Eq (18)]{scutari2010convex}, \cite[Proposition 1.3.4 and 3.2.1]{facchinei2007finite}. 
\end{proof}
In the following, we are going to focus on games for which the operator $F(x, y)$  is strongly monotone according to the following definition. Sufficient conditions for strong monotonicity to hold  in the case of NAGs are derived in Section \ref{sec:SMON}.

\begin{definition}\label{def:mon}
An operator $f(x):\mathcal{X}\subseteq \R^n\rightarrow \R^n$ is strongly monotone in $\mathcal{X}$ if there exists $\alpha>0$ such that $(f(x_1)-f(x_2))^\top(x_1-x_2)\ge \alpha \|x_1-x_2\|^2$ for all $x_1,x_2\in\mathcal{X}$. Equivalently, for all $x\in\mathcal{X}$, $\frac{\nabla_x f(x)+\nabla_x f(x)^\top}{2}\succeq \alpha I_n$, \cite[Proposition 2.3.2]{facchinei2007finite}. 
\end{definition}

\begin{assumption}[Strong monotonicity]\label{ass:SMON}
The operator $[x,y]\mapsto F(x, y)$ as defined in  \eqref{eq:F} is continuously differentiable in $[x,y]\in\mathcal{X}\times\mathcal{Y}$ and it is strongly monotone in $x\in \mathcal{X}$ for any fixed $y\in\mathcal{Y}$.
\end{assumption}

\subsection{ Main sensitivity analysis result for constrained games}
\label{theory} 

We here present our main sensitivity analysis result which holds for any game (i.e. not necessarily a NAG) with constraints satisfying the following assumption.

\begin{assumption}[Constraint qualification]\label{cq} Given a parameter $\bar y\in\mathcal{Y}$ and the corresponding Nash equilibrium $x^\star(\bar y)$ let $B^0, b^0$ be the matrices obtained by deleting from $B, b$ all the rows corresponding to constraints that are not active at $x^\star(\bar y)$ (i.e. the rows $k\in\N[1,m]$ such that $B_{(k,:)}x^\star(\bar y)< b_k$). Let $A:=[B^0 ; H]$ and $a:=[b^0,h]$. We assume that:
1) $A$ has full row rank;
2) the strict complementarity slackness condition
$\lambda_k(\bar y)>0 \mbox{ when } B_{(k,:)}x^\star(\bar y)= b_k $ 
is satisfied.\footnote{Under Assumptions \ref{constraints}, \ref{cost}  and  \ref{cq}.1, the dual variable $\lambda(\bar y)$ is unique, as shown in the proof of Theorem \ref{thm:sens}. }
\end{assumption}

We note that Assumption \ref{cq} is automatically satisfied if only equality constraints are present.

\begin{theorem}\label{thm:sens}
Suppose that Assumptions \ref{constraints},  \ref{cost} and \ref{ass:SMON}  are satisfied. Then for all $y\in\mathcal{Y}$ the game $\mathcal{G}(y)$ admits a unique Nash equilibrium $x^\star(y)$.  If additionally  Assumption \ref{cq} holds for $\bar y\in\mathcal{Y}$ then $x^\star(y)$ is locally differentiable at $\bar y$ and 
\begin{equation}\label{sensitivity}
\begin{aligned}
\nabla_y x^\star(\bar y) &= - M [\nabla_y F(x,y) ]_{\{x=x^\star(\bar y),y=\bar y\}}\\
\end{aligned}
\end{equation}
where 
\begin{equation}\label{L}
\begin{aligned}
M&:=[L-LA^\top [ALA^\top]^{-1}AL],\\ L&:=[\nabla_x F(x,y) ]_{\{x=x^\star(\bar y),y=\bar y\}}^{-1}
\end{aligned}
\end{equation}
and $A$ is as defined in Assumption \ref{cq}. Moreover, $M\succeq 0$.
\end{theorem}

\begin{proof}
The fact that  under Assumptions \ref{constraints}, \ref{cost} and \ref{ass:SMON} the Nash equilibrium $x^\star(y)$ of the game $\mc{G}(y)$ exists and is unique comes from existence and uniqueness of the solution of VI$(\mathcal{X},F(\cdot,y))$, \cite{facchinei2007finite}.
Set $w(y)=[x^\star(y); \lambda(y);\mu(y)]$ where $\lambda(y),\mu(y)$ are the multipliers associated with $x^\star(y)$, as in \eqref{KKT}. The fact that, under Assumptions \ref{constraints}, \ref{cost},  \ref{ass:SMON} and \ref{cq}, $w(y)$ and thus $x^\star(y)$ are  locally unique and  differentiable is known,  \cite[Theorem 7.16]{friesz2016foundations}.
The formula for the derivative of the  vector $w(y)$  provided therein is
$$\nabla_y w(\bar y)=-[\nabla_w \Gamma(w,y)]^{-1}[\nabla_y \Gamma(w,y)]\mid _{\{w=w(\bar y), y=\bar y\}},$$
where $\Gamma(w,y)$ is the KKT system obtained by considering only the equalities in \eqref{KKT1} and \eqref{KKT2}, see  \cite[Eq. (7.126)]{friesz2016foundations} for more details. Note that in general the formula for $\nabla_y x^\star(\bar y)$ obtained from the first block component of $\nabla_y w(\bar y)$ depends on the dual variables $\lambda(\bar y),\mu(\bar y)$.
By considering only the active constraints instead of the whole KKT system we here derive   an equivalent formula for $\nabla_y x^\star(\bar y)$ that does not depend explicitly on such dual variables. This primal reformulation allows us to derive the  results discussed in the next sections.
To this end, we note that if Assumption~\ref{cq} holds for a fixed $\bar y\in\mathcal{Y}$ then, by Theorem 7.16 in \cite{friesz2016foundations}, 
  in a neighbourhood of $\bar y$, the set of active constraints is unchanged, strict complementarity holds and the active constraints can thus  be rewritten as $Ax=a$. Consequently, a vector $x^\star(y)$ is a Nash equilibrium in a neighbourhood of $\bar y$ if and only if there exists a multiplier $\eta(y):=[\lambda_{>0}(y); \mu(y)]$ such that the following reduced KKT system holds
\begin{equation}\label{KKT_local}
\begin{aligned}
F(x^\star(y),y)+A^\top\eta(y)=0\\
Ax^\star(y)-a=0
\end{aligned}
\end{equation}
where $\lambda_{>0}(y)$ is the subvector of $\lambda(y)$ obtained by selecting only the components that are strictly positive.
Since $w(y)$ is differentiable in $y$, we can  differentiate \eqref{KKT_local} with respect to $y$ 
thus obtaining
\begin{align*}\label{KKT_local_derivative}
\nabla_x F(x^\star(y),y) \nabla_y x^\star(y) + \nabla_y F(x^\star(y),y) +A^\top \nabla_y  \eta(y)=0\\
A \nabla_y x^\star(y)=0
\end{align*}
From the first equivalence it follows
\begin{equation}\label{step1}
\nabla_y x^\star(y) = -L[ \nabla_y F(x^\star(y),y) +A^\top \nabla_y  \eta(y)].
\end{equation}
Note that $L$ is well defined since $F(x,y)$ strongly monotone implies that $\nabla_x F(x,y)$ is positive definite and thus invertible.
Substituting \eqref{step1} in the second equation we get
$$ \nabla_y  \eta(y)=- (ALA^\top)^{-1}[AL \nabla_y F(x^\star(y),y)]$$
and finally substituting back into \eqref{step1}, we get formula \eqref{sensitivity}.  Note that $ALA^\top$ is invertible since $A$ has full row rank by Assumption \ref{cq}.1.
By Shur complement $M\succeq0$ iff the matrix
$\begin{ps} L & LA^\top \\ AL & ALA^\top
\end{ps} = \begin{ps} I \\ A
\end{ps}  L \begin{ps} I & A^\top
\end{ps} \succeq 0.$
The last inequality holds since $L$ is the inverse of a positive definite matrix.\end{proof}

\section{Sufficient condition for strong monotonicity in NAGs}
\label{sec:SMON}

In this section we consider a fixed  parameter  $y$ and we  study the properties of the mapping $x\mapsto F(x,y)$ associated with the  specific NAG $\mathcal{G}(y)$.
We note that continuous differentiability of $F(x,y)$ is a trivial consequence of Assumptions \ref{constraints} and~\ref{cost}.

Our second main result is to derive two sufficient conditions, in terms of  game primitives only, for the operator of a NAG to be strongly monotone, thus guaranteeing Assumption \ref{ass:SMON}.
To this end, let us define
\begin{align}\label{k1}
\kappa_1(x,y)&:=\min_{i}\quad\ \   \lambda_{\textup{min}}( \nabla^2_{x^i}J^i(x^i,z^i,y^i)\mid_{z^i=z^i(x)} )\\ \label{k2}
\kappa_2(x,y)&:=\max_{i}\quad  \|\nabla^2_{x^iz^i}J^i(x^i,z^i,y^i)\mid_{z^i=z^i(x)} \|,
\end{align}
where $\nabla^2_{x^i}J^i(x^i,z^i,y^i)$ is the Hessian of the cost function of player $i$ with respect to $x^i$ while $\nabla^2_{x^iz^i}J^i(x^i,z^i,y^i)$ is the Jacobian of $\nabla_{x^i} J^i(x^i,z^i,y^i)^\top$ with respect to $z^i$.
Note that by Assumption \ref{cost}, the Hessian $\nabla^2_{x^i}J^i(x^i,z^i,y^i)\mid_{z^i=z^i(x)} \succ 0$ hence $\kappa_1(x,y)>0$ for all $x$ and $y$.
We state the two sufficient  conditions in the next assumption.
\begin{assumption}[Conditions for strong monotonicity]\label{mon_suff}
Consider a fixed $y\in\mathcal{Y}$. Suppose that
\begin{align}\label{alpha}
\alpha(y):=\min_x (\kappa_1(x,y) - \kappa_2(x,y) w(P) )>0
\end{align}
where
\begin{enumerate}
\item $w(P)=|\lambda_{\textup{min}}(P)|$ if the following  conditions hold\footnote{
Note that under condition $1$-a) $P$ is a symmetric, non-negative matrix with zero diagonal. Since $P=P^\top $ all its eigenvalues are real.
By Perron-Frobenius theorem $\lambda_{\textup{max}}(P)=\rho(P)>0$. Moreover,  since $P$ has zero diagonal, $\mbox{Tr}(P)=\sum_i \lambda_i(P)=0$. Consequently, it must be $\lambda_{\textup{min}}(P)<0$.
}
\begin{enumerate}
\item $P=P^\top$
\item   $\nabla^2_{x^iz^i}J^i(x^i,z^i,y^i)\mid_{z^i=z^i(x)}=\kappa_2(x,y) I_n$ for all $i$, for all $x\in\mathcal{X}$.
\end{enumerate} 

\item 
$ w(P)=   \|P\| $, otherwise.
\end{enumerate}
\end{assumption}

\begin{theorem}[Continuity and strong monotonicity]\label{thm:mon}
If  Assumptions \ref{constraints} and~\ref{cost} hold then the operator $x\mapsto F(x,y)$ as defined in \eqref{eq:F} is continuously differentiable in $x$.  If additionally  Assumption \ref{mon_suff} holds for a given  $y\in\mathcal{Y}$, then  $x\mapsto F(x,y)$  is strongly monotone in $\mathcal{X}$, with monotonicity constant $\alpha(y)$. 
\end{theorem}
\begin{proof} By assumption, $[x^i;z^i]\mapsto J^i(x^i,z^i,y^i)$ is twice continuously  differentiable hence $[x^i;z^i]\mapsto \nabla_{x^i}J^i(x^i,z^i,y^i)$ is continuously  differentiable in $[x^i;z^i]$. Note that $z^i(x)=\sum_{j\in\mathcal{N}^i} P_{ij} x^j$ is a linear function of $x$. Consequently, $x\mapsto \nabla_{x^i} J^i(x^i,\q^i(x),y^i)$ is continuously differentiable in $x$. This proves that $x\mapsto F(x,y)$ is continuously differentiable for $x\in\mathcal{X}$. 
To prove strong monotonicity, we start by noticing that the operator $\nabla_x F(x,y)$ can be rewritten as follows
\begin{align}\label{eq:grad_F}
\nabla_x F(x,y)= D(x,y)+ K(x,y) G,
\end{align}
where 
\begin{align}
D(x,y)&:=\mbox{blkdiag}[\nabla^2_{x^i}J^i(x^i,z^i,y^i)\mid_{z^i=z^i(x)}]_{i=1}^N \\
K(x,y)&:=\mbox{blkdiag}[ \nabla^2_{x^iz^i}J^i(x^i,z^i,y^i)\mid_{z^i=z^i(x)} ]_{i=1}^N \\
G&:=P\otimes I_n.
\end{align}
Note that by the properties of the Hessian $D(x,y)=D(x,y)^\top$. Moreover by definition \eqref{k1},  $D(x,y)\succeq \kappa_1(x,y) I$ for all $x\in\mathcal{X}$ and all $y\in\mathcal{Y}$.  By \eqref{eq:grad_F}  it then holds
\begin{multline*}
\textstyle \frac{\nabla_x F(x,y)+\nabla_x F(x,y)^\top}{2}\succeq   \kappa_1(x,y) I +\frac{K(x,y) G+ G^\top K(x,y)^\top}{2}.
\end{multline*}

Note that
\begin{itemize}
\item under Assumption \ref{mon_suff}.1)  the matrices $K(x,y)$ and $G$ are symmetric, moreover $K(x,y)=\kappa_2(x,y) I_{Nn}\succeq 0$. 
By assumption, $P$ is a non-negative, symmetric matrix with zero diagonal. It follows that $\lambda_{\textup{min}}(P)<0$ and $\lambda_{\textup{max}}(P)=\rho(P)$, as detailed in the footnote $2$. Moreover, by the properties of the Kronecker product,   $\Lambda(G)=\Lambda(P)$.  Hence $\lambda_{\textup{min}}(G)=\lambda_{\textup{min}}(P)=-|\lambda_{\textup{min}}(P)|$. Consequently,

\begin{equation}\label{eq:key_step}
\begin{aligned}
&\textstyle \lambda_{\textup{min}}\left(\frac{K(x,y) G+ G^\top K(x,y)^\top}{2}\right)= \kappa_2(x,y) \textstyle \lambda_{\textup{min}}\left(\frac{G+ G^\top }{2}\right)\\
 &= \kappa_2(x,y) \textstyle \lambda_{\textup{min}}(G)=-|\lambda_{\textup{min}}(P)| \|K(x,y)\|.
\end{aligned}
\end{equation}

\item Under Assumption \ref{mon_suff}.2)
\begin{align*}
&\lambda_{\textup{min}}\left(\frac{K(x,y) G+ G^\top K(x,y)^\top}{2}\right)\\ &\ge -\rho\left(\frac{K(x,y) G+ G^\top K(x,y)^\top}{2}\right)  \\
&= - \|\frac{K(x,y) G+ G^\top K(x,y)^\top}{2}\|\\
&\ge - \frac{\|K(x,y) G\|+ \|G^\top K(x,y)^\top\|}{2} \\
&= - \|K(x,y) G\| \ge  - \| P\| \|K(x,y)\| .\end{align*}
where we used the fact that for symmetric matrices the spectral radius equals the $2$-norm and $\|G\|=\|P \otimes I_n\|=\|P\|$.
\end{itemize}
Hence in all cases it holds
\begin{align*}
\lambda_{\textup{min}}\left(\frac{K(x,y) G+ G^\top K(x,y)^\top}{2}\right)\ge - w(P)\|K(x,y)\| .
\end{align*}
Note that 
$$\|K(x,y)\| \le  \max_{i} \|\nabla^2_{x^iz^i}J^i(x^i,z^i,y^i)\mid_{z^i=z^i(x)}\|  =  \kappa_2(x,y), $$
since the norm of a block diagonal matrix equals   the largest norm of its blocks.
Using the last result we can immediately see that under Assumption \ref{mon_suff}
 \begin{align*}
&\textstyle \lambda_{\textup{min}}(\kappa_1(x,y) I + \frac{K(x,y) G+ G^\top K(x,y)^\top}{2}) \\&\textstyle = \kappa_1(x,y) + \lambda_{\textup{min}}(\frac{K(x,y) G+ G^\top K(x,y)^\top}{2}) \\&\ge \kappa_1(x,y) - \kappa_2(x,y) w(P) \ge \alpha(y).
\end{align*}
Hence \begin{equation}\label{eq:inter} \begin{aligned}
\textstyle \frac{\nabla_x F(x,y)+\nabla_x F(x,y)^\top}{2} &\textstyle \succeq \lambda_{\textup{min}}( \kappa_1(x,y) I +\frac{K(x,y) G+ G^\top K(x,y)^\top}{2}) I \\&\succeq \alpha(y) I.
\end{aligned} 
\end{equation}
Equation \eqref{eq:inter} implies that $F(x,y)$ is strongly monotone with constant $\alpha(y)$.
\end{proof}

\section{ Quadratic network games}
\label{systemic_risk}

\subsection{ The setting}

Consider the specific class of NAGs with cost functions 
\begin{equation}\label{nag_q}
J^i(x^i,\q^i(x),y^i)= \frac{1}{2} (x^i)^\top I_n x^i - f(\q^i(x)+ y^i)^\top x^i,
\end{equation}
for some $f:\R^n\rightarrow \R^n.$
Games with such cost functions are known in the literature under the name of quadratic network games \cite[Section 2.1]{acemoglu2015networks} (or linear-quadratic network games if additionally the interaction function is linear \cite[Section 4.1]{jackson2014games}).
For such class of games, the parameter $y^i\in\mathcal{Y}^i\subseteq \R^n$ usually represents a localized shock or perturbation to the payoff of agent $i$. In the following, we  assume that these shocks $y^i$ are small, independently and identically distributed and have mean $\hat y$.

If the strategy of the agents are \textit{unconstrained} (i.e. $\mathcal{X}^i=\R^n$) and $P_{ii}=0$ for all $i\in[1,N]$ then each agent's best response can be characterized explicitly as
$$x^i_{\textup{br}}(\q^i(x),y^i)=f(\q^i(x)+ y^i).$$
By using this reformulation, the recent work  \cite{acemoglu2015networks} conducted sensitivity analysis of the Nash equilibrium and studied volatility of aggregates to shocks (perturbations) in unconstrained quadratic games. We briefly recap such results in Section \ref{unconstrained} and show that they can be recovered as a special case of our main formula \eqref{sensitivity}.

In Section \ref{constrained} we exploit formula \eqref{sensitivity} to extend the results of \cite{acemoglu2015networks} to the  \textit{constrained} case. In this setting, constraints emerge, for example, when a central authority has an objective function that depends on the exact structure of the network (such as minimizing the difference between the actions of a subset of agents over  part of the network or keeping the actions or beliefs of certain agents unchanged).  Following \cite{acemoglu2015networks},   we model the \textit{aggregate  output of the system} as
$$s(y)=g(h(x^{\star 1}(y))+\ldots+h(x^{\star N}(y))),$$
where $h:\R^n\rightarrow \R^n$ and $g: \R^n\rightarrow \R$.
Note that, in the absence of shocks the Nash equilibrium is $x^\star(0)$ with associated output $s(0)$.  Suppose now that the central authority can intervene and guarantee that, no matter what the realization of the shocks,  one agent does not change its strategy.\footnote{We focus on one agent for simplicity. Similar results can be derived also for the case when the central authority can influence more than one agent.} This intervention can be modeled by  adding the constraint $x^{k}=x^{\star k}(0)$ to the game. Note that this constraint does not modify the unperturbed system output, but modifies how the system reacts to shocks. Our objective is to address the following question: \textit{Which agent should be constrained by the central authority in order to minimize the effect of the shocks on the system output $s(y)$ at the new equilibrium?} By exploiting formula \eqref{sensitivity} we  provide an answer to this question both from an ex-post perspective (i.e., if the central authority knows the actual realization $y$ of the shock) and from an ex-ante perspective (i.e., if the authority must commit to his decision before the shock is realized).

In the following we focus our discussion on the scalar case, formula \eqref{sensitivity} however allows an immediate generalization of our results to the multidimentional case.  We  also assume the following regularity conditions and normalizations.
\begin{assumption}(Network games)\label{network} The adjacency matrix $P$ of the network is non-negative and, 
for each $i\in\N[1,N]$, $P_{ii}=0$.  Moreover, $f(0)=g(0)=h(0)=0$. For each $w\in \R$, $0\le f'(w)\le\gamma$ and $\gamma \|P\|<1$. 
\end{assumption}
Note that the previous assumption  implies that the game is normalized so that, in the absence of shocks, the equilibrium state $x^{\star}(0)$ and the  output of the system $s(0)$ are both zero. Moreover, Theorem \ref{thm:mon} guarantees that under Assumption \ref{network}, Assumption \ref{ass:SMON} is satisfied.

\subsection{The  unconstrained case}\label{unconstrained}

In \cite{acemoglu2015networks} the following formula for  the Nash sensitivity  in unconstrained quadratic network games is derived 
\begin{align}\label{sensitivity_scalar}
&\nabla_{y^i} x^\star(0)= f'(0) [I-f'(0)P]^{-1}e_i =:f'(0)L e_i, \\
&\mbox{ so that }  \quad
\textstyle \frac{\partial x^{\star j}(0)}{ \partial y^i}  =f'(0)L_{ji}.
\end{align}
The matrix $L:=[I-f'(0)P]^{-1}$ is known in the economic literature as  \textit{Leontief matrix}.
It follows from Assumption~\ref{network}  that $L$ is well defined (i.e. the inverse exists). 
By using  Neumann series it is also easy to see that $L=[I-f'(0) P]^{-1}=\sum_{l=0}^\infty f'(0)^l P^l$. Consequently, the element $L_{ji}$ is the sum over all the possible paths from $i$ to $j$ of the product of the edge weights in the path  discounted by the factor $f'(0)^l$, where $l$ is  the path length. Intuitively  the element $L_{ji}$ describes the ways in which a shock applied to agent $i$  propagates to agent~$j$. It is then immediate to see that $L$ has nonnegative entries and that $L_{ii}\ge 1$ (take $l=0$).

The output sensitivity to a shock to agent $i$  derived in  \cite{acemoglu2015networks} is
\begin{equation}\label{sens_out_unc}
\textstyle \frac{\partial s(0)}{\partial y^i}=f'(0)g'(0)h'(0)\sum_{m=1}^N L_{mi}=:\alpha v^i
\end{equation}
where $\alpha:=f'(0)g'(0)h'(0)$. 
It is evident from \eqref{sens_out_unc} that whether a  positive shock $y^i$ to agent $i$ locally increases or decreases the output depends only on the sign of  $\alpha$ or in other words on the convexity/concavity pattern of the functions $g,h$ at the origin  \cite{acemoglu2015networks}. The strength of such effect, on the other hand, depends  on the parameter $v^i:=\sum_{m=1}^N L_{mi}$, which is the  Bonacich centrality of agent $i$.  

\begin{definition}[Bonacich centrality]
The element
$\textstyle v^i:=\sum_{m=1}^N L_{mi},$
obtained by summing  the entries of the Leontief matrix in column $i$, is the  \textit{Bonacich centrality} of agent $i$. 
\end{definition}

 Note that  $v^i$ is always positive and corresponds to the weighted sum of the paths through which the shock can propagate from agent $i$ to any other agent in the network.

\subsection{ The  constrained case}\label{constrained}

Theorem~\ref{thm:sens}  allow us to generalize  formula \eqref{sensitivity_scalar} to the constrained case. Specifically, under Assumptions \ref{constraints} to~\ref{cq} we obtain
\begin{equation}\label{sensitivity_scalar_constrained}
\begin{aligned}
\nabla_{y^i} x^\star(0)&=  f'(0) [L-LA^\top [ALA^\top]^{-1}AL] e_i\\
&=  f'(0) Le_i -f'(0)LA^\top [ALA^\top]^{-1}ALe_i, 
\end{aligned}
\end{equation}
where $L$ is  the same Leontief matrix of equation \eqref{sensitivity_scalar}. 
By comparing  \eqref{sensitivity_scalar} and \eqref{sensitivity_scalar_constrained}, one can immediately see that the presence of the constraints modifies the sensitivity  in an additive fashion. Similarly, the output sensitivity becomes
\begin{equation}\label{sensitivity_output_constained}
\textstyle \frac{\partial s(0)}{\partial y^i}=\alpha [v^i -  \mathbbm{1}^\top LA^\top [ALA^\top]^{-1}AL e_i ].
\end{equation}
We then see that the unconstrained case (i.e. when the $A$ matrix is empty) is just a special case of \eqref{sensitivity_scalar_constrained} and
\eqref{sensitivity_output_constained}.

We next consider the  case when the central authority adds the specific constraint $x^k=x^{\star k}(0)=0$, as motivated in the introduction of this section.   Note that since this is an equality constraint Assumption \ref{cq} is always  met and $A=e_k^\top$. If agent $k$ is constrained, a shock to an arbitrary agent $i$ affects an agent $j$ according to 

$$ \frac{\partial x^{\star j}(0 \mid k)}{\partial y^i}=  f'(0) \left[L_{ji}-\frac{L_{jk} L_{ki}}{L_{kk}}\right]. $$
Note that  the terms $L_{ji}$ and $\frac{L_{jk} L_{ki}}{L_{kk}}$ have always the same sign, therefore fixing the strategy of an agent has the effect of reducing the sensitivity to the shock for the others. The relevance of such effect depends on the weight of the paths from $i$ to $k$ and from $k$ to $j$, normalized by the weight of the  loops passing through $k$. The output sensitivity becomes 
\begin{equation}\label{sen_out}
\frac{\partial s (0\mid k)}{\partial y^i}=\alpha [\sum_{m=1}^N L_{mi} -  \sum_{m=1}^N  \frac{L_{mk} L_{ki}}{L_{kk}}  ]=:\alpha[v^i-v^i_k]
\end{equation}
where we define $v^i_k:=v^k\frac{L_{ki}}{L_{kk}}$.
\begin{lemma}\label{lem:vi}
For all $i,k\in\N[1,N]$, $v^i\ge v^i_k$.
\end{lemma}
\begin{proof}
Since $L$ is a Leontief matrix, for any $i,k,j\in\N[1,N]$ by \cite{zeng2001property} we have  
\begin{equation}
\textstyle L_{mi}L_{kk}-L_{mk}L_{ki}\ge 0 \Rightarrow L_{mi}\ge\frac{L_{mk}L_{ki}}{L_{kk}},
\end{equation}
where we used the fact that $L_{kk}\ge 1$.
By summing over all $m\in\N[1,N]$ we obtain the desired result. 
\end{proof}
Therefore constraining agent $k$ attenuates the sensitivity of the output but can never reverse its sign, which again depends only on the sign of $\alpha$.

By using \eqref{sen_out} we can approximate the ex-post output, given that agent $k$ was constrained, as
$$\textstyle s(y\mid k) \approx s(0)+ \sum_{i=1}^N \frac{\partial s(0\mid k)}{\partial y^i}y^i= \sum_{i=1}^N \alpha  [v^i-v^i_k] y^i.$$

In the following we are going to assume that $\alpha>0$ and shocks are non-negative $y^i\ge 0$. Similar results can be derived for the other cases. Under these assumptions and given Lemma~\ref{lem:vi}, $s(y\mid k) \ge 0$ for all $k\in\N[1,N]$. Since $s(0)=0$, to minimize the effect of the shocks
the central authority should constrain agent
$$\textstyle \bar k_{\textup{post}}:=\!\arg\!\min_k \alpha \sum_{i=1}^N   [v^i-v^i_k]  y^i =\! \arg\!\max_k \!  \sum_{i=1}^N \! v^i_k y^i. $$
Note that, to use this formula, the central authority  needs to know the realization $y$ of the shock.
If, on the other hand, the central authority needs to commit to a decision ex-ante then it should minimize the expected outcome

$$\textstyle \mathbb{E}_y[s(y\mid k)] \approx \sum_{i=1}^N \alpha  [v^i-v^i_k] \hat y,$$
where we used $\mathbb{E}[y^i]=\hat y$ for all $i\in\N[1,N]$, and thus constrain the opinion of agent 
\begin{align*}
\bar k_{\textup{ante}}&:=\textstyle\!\arg\!\min_k \quad \sum_{i=1}^N  \alpha  [v^i-v^i_k] \hat y =\! \arg\!\max_k \quad \! \sum_{i=1}^N \! v^i_k \\&=:\textstyle \arg\!\max_k \quad w^k .
\end{align*}
The quantity $w^k$ is  known as \textit{inter-centrality}  or \textit{key-player centrality}.
\begin{definition}[Key-player centrality]
The element
\begin{align}\label{w}
\textstyle w^k:=\sum_{i=1}^N v^i_k= \sum_{i=1}^N v^k\frac{L_{ki}}{L_{kk}} = \frac{ v^k }{L_{kk}} \sum_{i=1}^N L_{ki},
\end{align}
where $L$ is the Leontief matrix, is the \textit{inter-centrality}  or \textit{key-player} centrality of agent $k$.
\end{definition}
This centrality measure was originally derived in \cite{ballester2006s}, by using a   different approach, to identify the key-player to be removed in an 
 unconstrained quadratic-linear network game with strategic complements to maximize the  decrease in aggregated output at equilibrium.
 The application considered here is a different example of where such a measure can be applied. To sum up, the optimal strategy for the central authority is to target the agent with maximum {inter-centrality} and not the one with maximum Bonacich centrality, contrary to what one could have expected. 
 
The example given in Figure \ref{intuition} provides intuition for the difference between the two measures.\footnote{A special case when the ordering according to the two centrality  measures coincides is the case of a symmetric network (i.e. $L_{ki}=L_{ik}$) with no loops ($L_{kk}=1$). In that case
$
\textstyle w^k= \frac{ v^k }{L_{kk}} \sum_{i=1}^N L_{ik}  =\frac{ (v^k)^2 }{L_{kk}}=(v^k)^2.
$ In general, however, the ordering is different, as illustrated in Fig.~\ref{intuition}. } Therein we consider, for different values of $f'(0)$, a directed network without loops where all the edges have the same weight. Note that in the absence of loops $L_{ii}=1$ for each agent $i$, hence 
 $$v^i=\sum_{m=1}^N L_{mi},\quad  w^i= \left(\sum_{m=1}^N L_{mi}\right)\left(\sum_{m=1}^N L_{im}\right).$$
 Recall that $L_{mi}$ is the total weight of paths from node  $i$ to node $m$, discounted by $f'(0)^l$, where $l$ is the path length. 
 A first thing to note is that the Bonacich centrality $v^i$ only depends on \textit{out-going}  paths, while the inter-centrality $w^i$ depends on the product of both \textit{in-coming} and \textit{out-going} paths. This explains why the latter is the relevant measure to be considered in deciding which agent should be constrained. In fact, constraining  an agent has two effects. The first \textit{direct effect} is that the strategy of the agent influences all its followers, the second \textit{indirect effect} is that it blocks all the shocks coming from previous nodes and prevents then from spreading further over the network. It is then clear that the relevance of an agent as an influencer (direct effect) depends on its out-going paths, but its relevance as a blocker (indirect effect) depends on the in-coming paths. One should then aim at constraining an agent for which both  effects are relevant.
 A second thing to note is that how strong these effects are depends on how strong the  interaction is among the  agents (i.e. how large $f'(0)$ is).  From the interpretation of $f'(0)$ as a discount factor in $L_{mi}$ we see that, when $f'(0)$ is very small, only short paths are important. In the limit, for very small  $f'(0)$, the Bonacich centrality of an agent $i$ is its out-degree, while the inter-centrality is the product between out- and in-degree.\footnote{To be rigorous we should say the in-degree plus one and the out-degree plus one, to account for $L_{ii}=1$.} It is then clear why in Figure  \ref{intuition}, for the case $f'(0)=0.1$, the agent with maximum Bonacich centrality is the one on the right (with out-degree $6$) while there are two agents with same maximum inter-centrality (with product of out- and -in degree equals to $1\cdot 6$). On the other hand, the higher $f'(0)$ is, the more paths of higher length matters. This explains why for $f'(0)=1$ the agents with highest Bonacich centrality are the ones on the extreme left (i.e. the ones from which more paths are departing). Note that the extreme agents can never have high inter-centrality because they always have either very low in-path weight or out-path weight.
 
 \begin{figure}
\begin{center}
  \includegraphics[width=0.5\textwidth]{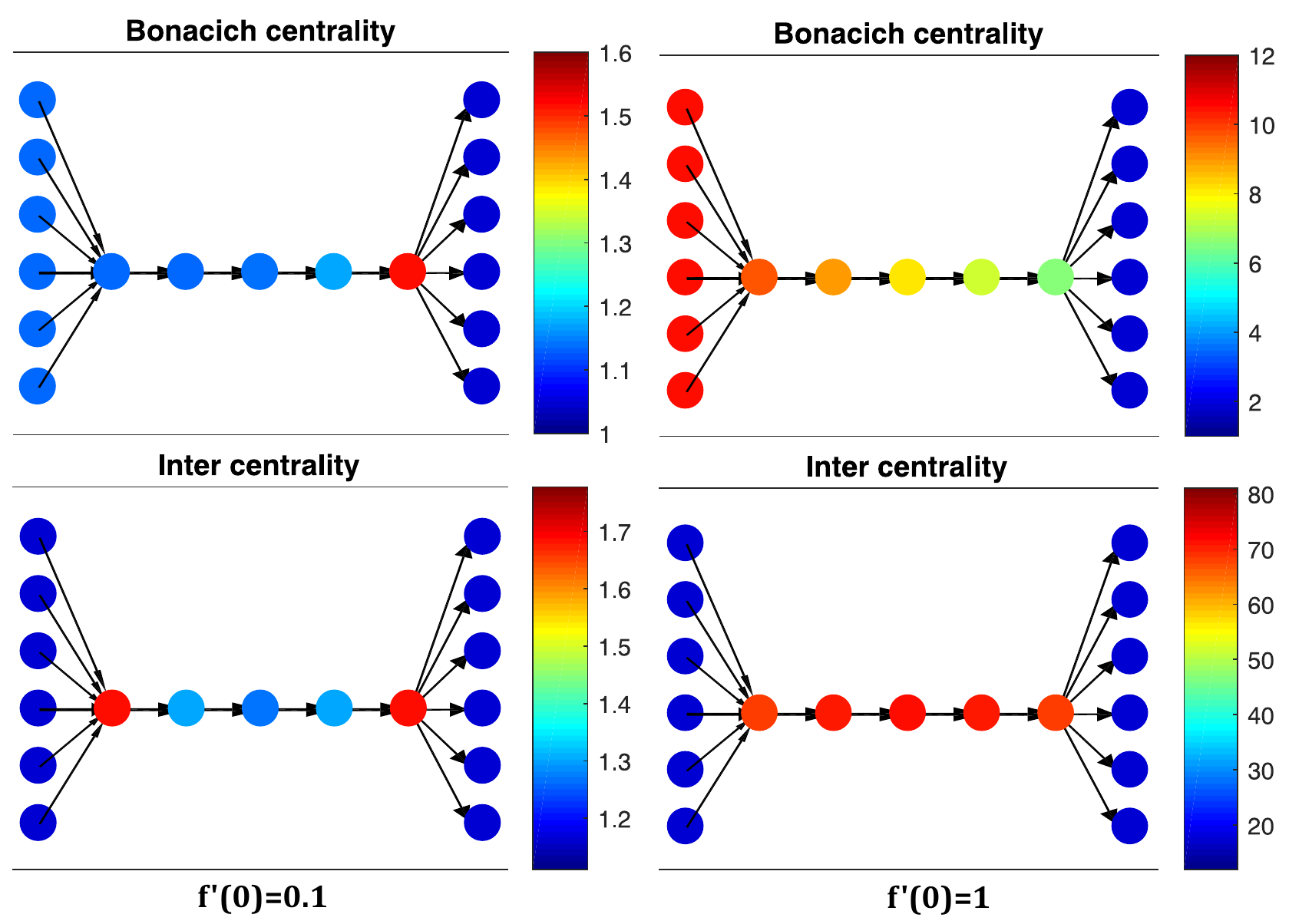}
 \end{center}
\caption{\small  Plot illustrating the difference between Bonacich and inter-centrality measures for different values of $f'(0)$. The node color is an indicator of  the centrality of the agents.}
\label{intuition}
\end{figure}

 While in this discussion we focused on the case when the central authority can constrain the opinion of a single agent, formula \eqref{sensitivity_output_constained} can  be used to perform the same analysis for any arbitrary number $K$ of agents, by setting $A=[e_{k_j}^\top]_{j=1}^K$.

\subsection{ Rumour propagation in opinion dynamics}
\label{opinion_dynamics_full}
As application of our theoretical results we 
consider a network where agents exchange and update  opinions regarding a certain topic according to the Friedkin and Johnsen \cite{friedkin1999social} model of opinion dynamics
\begin{equation}\label{fj}
\textstyle x^i_{(t+1)}= \frac{1}{1+\theta^i} \sum_{j=1}^N P_{ij} x^j_{(t)} +  \frac{\theta^i}{1+\theta^i} y^i.
\end{equation}
Here $x^i_{(t)}\in[0,1]$ denotes the opinion of agent $i$ at time $t$, $ y^i= x^i_{(0)}$ denotes its initial opinion and
 $\theta^i>0$ is a parameter that captures its  stubbornness (i.e., the weight that agent $i$ places on its initial opinion).  It was shown in \cite{ghaderi2014opinion} that if $P$ is non-negative, row stochastic, $P_{ii}=0$  and $\theta^i>0$ for all $i\in\N[1,N]$ the Friedkin and Johnsen  dynamics  eventually converge to the Nash equilibrium $x^\star(y)$ of the NAG\footnote{Note that if $P$ is row stochastic and the initial conditions are in the interval $[0,1]^N$ then under the dynamics in \eqref{fj} all the opinions remain between $[0,1]^N$. Consequently, the Nash equilibrium of \eqref{nag_opinion} is always in $[0,1]^N$  without the need of explicitly adding this constraint in the game.} with matrix $P$ and  cost 
\begin{equation}
\label{nag_opinion}
J^i(x^i,\q^i(x),y^i)=\|x^i-\q^i(x)\|^2+\theta^i \|x^i-y^i\|^2,
\end{equation}
where $z^i(x)=\sum_{j=1}^N P_{ij}x^j$.
Clearly,  for $y:=[y^i]_{i=1}^N=0$ all the agents have the same opinion at the Nash equilibrium which is $x^{\star i}(0)=0$. We here define rumours  as perturbations of the initial opinion of  certain agents from $0$ to $y^i>0$, where the $y^i$ are i.i.d. with average $\hat y$. According to the previously described model, at convergence, the  effect of the rumours is to modify the vector of final opinions from $x^\star(0)=0$ to $x^\star(y)\ge0$. We can then quantify how much the rumours have spread at the end of the process by 
$$\textstyle s(y):=\sum_{j=1}^N x^{\star j}(y).$$
The aim of this subsection is to characterise the opinion of which agent  should the central authority constrain to $0$ in order to minimize $s(y)$. For simplicity we consider the case  $\theta^i=\theta>0$ for all $i\in\N[1,N]$ and $P$  doubly stochastic, so that $\|P\|\le \sqrt{\|P\|_\infty \|P\|_1}=1$.

Note that, up to constant terms that do not depend on $x^i$ and some positive scaling, the cost function in \eqref{nag_opinion} can be rewritten as
\begin{align} \label{nag_opinion2}
\textstyle J^i(x^i,\q^i(x),\tilde y^i)= \frac 12 \|x^i\|^2 -\frac{1}{(1+\theta)}(\q^i(x)+ \tilde y^i) x^i,
\end{align}
where the $\tilde y^i:=\theta y^i$ are i.i.d. with average $\theta \hat y$.
Then \eqref{nag_opinion2} belongs to the class of NAGs discussed in this section with  $f(x)=x/(1+\theta)$, $g(h)=h$, $h(x^i)=x^i$,  $\alpha=\frac{1}{1+\theta}>0 $ and $F(x,\tilde y)=(I-\frac{1}{1+\theta} P)x-\frac{1}{(1+\theta)} \tilde y$.
Note that $\nabla_x F(x,\tilde y) = (I-\frac{1}{1+\theta} P)$ and $f'(x)=\frac{1}{1+\theta}=:\gamma<1$ for all $x\in\mathcal{X}$.  Consequently, $\gamma \|P\| \le \frac{1}{1+\theta}  <1$. Assumptions 1, 2 and \ref{network} are thus met. Assumption \ref{cq} is always met because we consider only equality constraints.  By the previous discussion the central authority should then constrain the opinion of the agent $k$ with maximum \textit{inter-centrality}. Note that since $F(x,y)$ is linear both in $x,y$ this result holds globally (i.e. even for large rumors).

\section{ Atomic routing  games with information constraints}
\label{traffic_sec}

\subsection{ The setting}

As a second  application, we consider  routing  games in which  $N$ atomic agents select how to split their traffic across  the available routes in a traffic network, with the goal of minimizing their travel time, under the assumption that the latter depends on the overall congestion level of the  network.

Specifically, we consider a directed graph $(\mathcal{V},\mathcal{E})$ where each node $v\in\mathcal{V}$ corresponds to a location and each directed edge $e=(u,v)\in\mathcal{E}$ corresponds to a road connecting $u$ to $v\neq u$. We define $H$ to be the node-edge incidence matrix of the road network. Equivalently, given any edge $e=(u,v)\in\mathcal{E}$, $H_{k,e}=-1$ if $k=u$, $H_{k,e}=1$ if $k=v$ and $H_{k,e}=0$ otherwise.

We assume that each agent $i\in\N[1,N]$ has a \textit{non-negligible} (atomic)\footnote{See \cite{roughgarden2007routing} for a distinction between atomic and non-atomic routing games.}  amount of flow $\eta^i>0$ that he needs to assign to the available routes between an  origin node $o^i\in\mathcal{V}$ and a destination node $d^i\in\mathcal{V}$. Note that the origin-destination pair may be different for each agent. Moreover, different from standard routing games, we assume that agent $i$ knows only a subset $\mathcal{E}^i \subseteq \mathcal{E}$ of the roads. Again this set may be different for each agent, see Figure \ref{net} for an example.

\begin{figure}
\begin{center}
\includegraphics[width=0.35\textwidth]{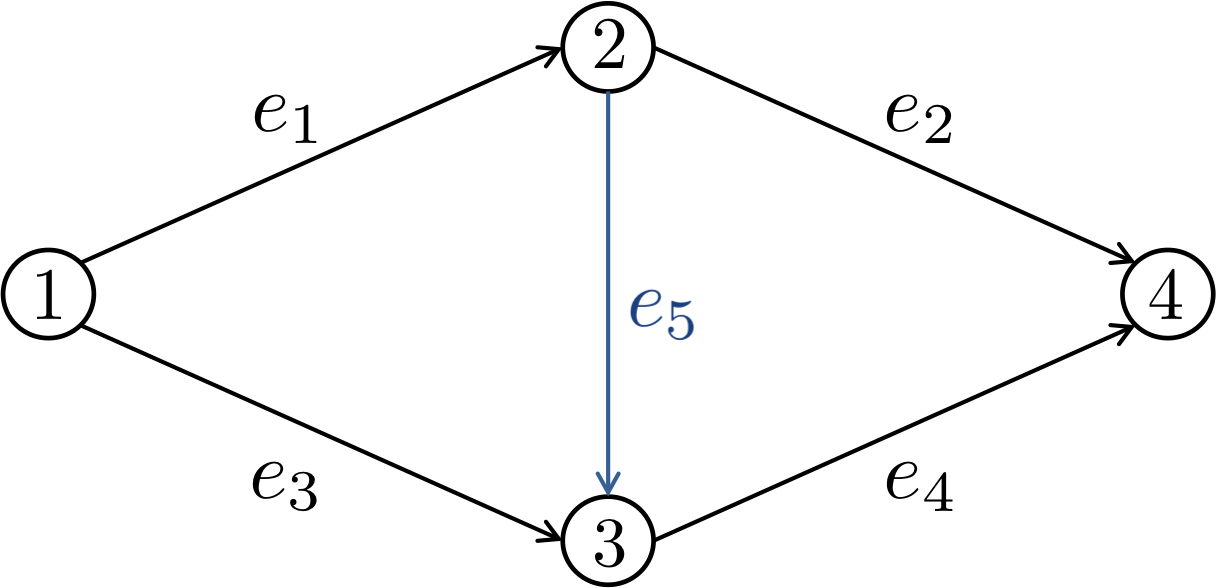}
\end{center}
\caption{\small Wheatstone road network. We  assume that  the agents have the same origin and destination, $o^i=1,d^i=4$ for all $i\in\N[1,N]$, but they have different sets of information. Specifically, some agents  know all the network (for them $\mathcal{E}^i:=\mathcal{E}:=\{1,2,3,4,5\}$) and some agents  do not know edge $5$ (for them $\mathcal{E}^i:=\{1,2,3,4\}\subset \mathcal{E}$). \vspace{-0.3cm}}
\label{net}
\end{figure}

To model these constraints let $x^i\in\R^{|\mathcal{E}|}$ be a vector whose component $x^i_e$ denotes the flow that agent $i$ is allocating on road $e$.  Then we say that a flow $x^i$ is feasible for agent $i$ if 
1) $x^i\ge 0$;
2) $Hx^i=h^i$, where $h^i$ is a vector whose entries are all $0$ except for the origin node $o^i$ that has entry $-\eta^i$ and the destination node $d^i$ that has entry $\eta^i$;
3) $x^i_e= 0$ for all $e\notin \mathcal{E}^i$.
Note that conditions 1) and 2) are the  feasibility conditions of the standard routing game, while condition 3) models the fact that the  agents might have different sets of information. 

Set $x:=[x^i]_{i=1}^N$ and let $\q(x):=\sum_{j=1}^N x^j$ be the vector of total edge flows. To model congestion effects we assume that the travel time for each edge is a  function $p_e(\q_e,y^e):\R_{\ge0}\times\R_{>0}\rightarrow \R_{>0}$ of the total edge flow $z_e$ which might depend on a parameter $y^e$ (i.e., number of lanes, speed limits, etc.). The travel time experienced by agent $i$ is therefore given by 
\begin{equation}\label{nag_traffic}
\textstyle J^i(x^i,\q(x),y):= \sum_{e=1}^{|\mathcal{E}|} p_e(\q_e(x),y^e) x^i_e.
\end{equation}
Note that in this application both the aggregate $z(x)$ and the parameter
 $y:=[y^e]_{e=1}^{|\mathcal{E}|}$ are the same for every agent. 
Since $z(x)$ depends on the sum of the flow vectors of all the agents this is a NAG with a matrix $P$ whose elements are all equal to one (not to be confused with the road network which appears in the constraints via the incidence matrix $H$). Our objective is to  study  the effect of the parameter $y$ on the total travel time 
\begin{equation} \label{totaltravel}
\begin{aligned}
s(y) &=\textstyle \sum_{e=1}^{|\mathcal{E}|}  p_e(\q_e(x^\star(y)),y^e) \q_e(x^\star(y))
\\&=: p(\q(x^\star(y)),y)^\top \q(x^\star(y))
\end{aligned}
\end{equation}
at the Nash equilibrium $x^\star(y)$ when information constraints such those  in 3) are present. Note that we defined $p(z,y):=[p_e(z_e,y^e)]_{e=1}^{|\mathcal{E}|}$. For convenience we rephrase here Definition~\ref{nash} for the case of  routing games.

\begin{definition}[Nash equilibrium]
A set of strategies $\{ x^{\star i}(y) \in \mathcal{X}^i\}_{i=1}^N$ is a Nash equilibrium for the routing game with parameter $y:=[y^e]_{e=1}^{|\mathcal{E}|}$ if for all players $i\in\N[1,N]$ and all strategies $x^i\in\mathcal{X}^i$ we have
\begin{equation}\label{traffic_nash}
\sum_{e=1}^{|\mathcal{E}|} p_e(\q^\star_e(y),y^e) x^{\star i}_e(y) \le \sum_{e=1}^{|\mathcal{E}|} p_e( x^i_e +\sum_{j\neq i}  x^{\star j}(y),y^e) x^i_e,
\end{equation}
where ${\q}^\star_e(y):= \sum_{j=1}^N  x^{\star j}_e(y)$ is the total edge flow and $$\mathcal{X}^i:=\{x^i\in\mathbb{R}^{|\mathcal{E}|}\mid x^i\ge0, Hx^i=h^i, x^i_e=0\ \forall e\notin \mathcal{E}^i\}$$ is the set of feasible edge flow allocations for agent $i$.
\end{definition}

 \subsection{ Sensitivity of the Nash equilibrium in routing games}
 \label{traffic_micro}
 
 As already noted, the atomic routing game is  a NAG.
The  operator associated with this game is 
\begin{equation}\label{operator_traffic_nash}
\begin{aligned}
F(x,y)&=[\nabla_{x^i} J^i(x^i,z(x),y)^\top]_{i=1}^N\\
&=[p(z(x),y) +\nabla_z p(z(x),y) x^i]_{i=1}^N.
\end{aligned}
\end{equation}
We assume the following regularity condition.
\begin{assumption}[Travel time functions]\label{traffic}
 Suppose that $F(x,y)$ in \eqref{operator_traffic_nash} satisfies Assumption \ref{cost} and \ref{ass:SMON}.  
\end{assumption} 
\begin{remark}
Note that for this application $P_{ii}=1\neq 0$, hence Theorem \ref{thm:mon} cannot be applied. Nonetheless, conditions on the travel time functions $p_e(\cdot,y^e)$ such that Assumption~\ref{traffic} is satisfied have been discussed for example in \cite{gentile2017nash}. As a special case we note that, Assumption~\ref{traffic} is always met in the case of affine and strictly increasing travel time functions. We also note that  Assumption \ref{constraints} is always satisfied in  routing games.
\end{remark}

 Under Assumption \ref{traffic},  Theorem \ref{thm:sens} guarantees the existence of a unique Nash equilibrium for the atomic routing game. To  compute  its sensitivity  we need to check Assumption~\ref{cq} at the current parameter $\bar y$. In the traffic context,  this condition ensures that, if an agent $i$ is not using road $e$ for the parameter $\bar y$ (i.e, $x^{\star i}_e(\bar y)=0$), then for small perturbations of the parameter it won't use road $e$ also in the new equilibrium $x^{\star i}(y)$. In other words, the subset of the constraints  $\{x^i_e\ge0\}_{e=1}^{|\mathcal{E}|}$ that are active should be locally unchanged (as guaranteed by  the proof of Theorem \ref{thm:sens} under Assumption  \ref{cq}).\footnote{This strict complementarity condition is similar to the strict  complementary condition commonly used in non-atomic routing games \cite{ lu2008sensitivity, do2014sensitivity}. }
If we  consider a parameter  $\bar y$ such that this assumption is met, then Theorem~\ref{thm:sens} allows us to compute  the Nash sensitivity  
\begin{equation}\label{sensitivity_traffic}
\begin{aligned}
\nabla_y x^\star( \bar y) &= - M [\nabla_y F(x,y) ]_{\{x=x^\star( \bar y),y= \bar y\}}\\
\end{aligned}
\end{equation}
where  $F(x,y)$ is as in \eqref{operator_traffic_nash} and $M,L$ are as defined in \eqref{L}. Note that, given $x^\star( \bar y)$,  the  matrix $A$ of active constraints used in \eqref{L} has the following structure  
$$A=\mbox{blkdiag}( \{ [H; R(x^{\star i}(\bar  y))] \}_{i=1}^N )$$ where $R(x^{\star i}(\bar y))$ is a  matrix constructed by adding one on top of each other  the set of  canonical vectors $\{e_e^\top| \forall e \in\mathcal{E} \mbox{ s.t. } x^{\star i}_e(\bar y)=0 \}$. Intuitively,  $R(x^{\star i}(\bar y))$ models the subset of the constraints  $\{x^i_e\ge0\}_{e=1}^{|\mathcal{E}|}$ that are active  at $x^{\star i}(\bar y)$.

 From \eqref{sensitivity_traffic}  one can then immediately derive the sensitivity of the total edge flow at the Nash equilibrium,
   \begin{align}\label{sensitivity_general}
  \nabla_y z^\star( \bar y)&= \textstyle \sum_{i=1}^N \nabla_y x^{\star i}(\bar y) 
= [ \mathbbm{1}_N^\top \otimes I_{|\mathcal{E}|} ]\nabla_y x^\star(\bar y)
  \end{align}
  and the sensitivity of the total travel time  
$s(y)$ in \eqref{totaltravel}
can immediately be computed from the sensitivity of the total edge flow as follows
\begin{equation}
\label{sensitivity_ttt}
 \nabla_y s(\bar y)= {\q^\star}^\top \nabla_y p({\q^\star},\bar y) +[p({\q^\star},\bar y)+\nabla_z p({\q^\star},\bar y){\q^\star}]^\top\nabla_y {\q^\star},
 \end{equation}
where we omitted the dependence of ${\q^\star}(\bar y)$ on $\bar y$ for simplicity. 
Formula \eqref{sensitivity_ttt} can be used to understand which road improvements lead, at least locally, to a higher improvement of travel time for the whole network. In fact 
$\textstyle s(y)\approx s(\bar y) + \sum_{e=1}^{|\mathcal{E}|} \frac{\partial s(\bar y) }{\partial y^e} (y^e- \bar y^e).$
 If we assume that a road improvement on road $e$ corresponds to a decrease of the parameter $y^e$, then $(y^e- \bar y^e)<0$ and the road that should be improved to minimize $s(y)$ is
$\textstyle \hat e:=\arg\max_{e\in \mathcal{E}} \frac{\partial s(\bar y) }{\partial y^e}.$
Note that it might happen that $\frac{\partial s(\bar y)}{\partial y^e}<0$ for some road $e\in\mathcal{E}$. In this case an improvement on road $e$ actually leads to an increase of total travel time. This pathological situation is known as \textit{Braess' paradox}. Formula  \eqref{sensitivity_ttt}  can then be used to check whether for a specific network the \textit{Braess' paradox}  occurs, as illustrated in the next subsection.\footnote{We finally note that it is possible to formulate the atomic routing game also in terms of   path flows  instead of edge flows  (as usually done in non-atomic games). In general  however    path flows  are not unique at the Nash equilibrium.
  Hence Theorem \ref{thm:sens} cannot be applied to this formulation.}

\subsection{ Simulation example}
To illustrate the  approach derived in Section \ref{traffic_micro} we consider the well known  Wheatstone $5$ roads network illustrated in Figure~\ref{net} with congestion functions 

\begin{align*}
 p(z,y)&=\arraycolsep=0.1pt 
\medmuskip = 0.1mu  \left[\begin{array}{c} p_1(z_1) \\p_2(z_2) \\p_3(z_3) \\p_4(z_4) \\ p_5(z_5,y)\end{array}\right] \arraycolsep=0.1pt
\medmuskip = 0.1mu  =\left[\begin{array}{ccccc} \frac{40}{150} & 0 & 0 & 0 & 0 \\0 &  \frac{1}{150} & 0 & 0 & 0 \\0 & 0 &  \frac{40}{150} & 0 & 0 \\0 & 0 & 0 &  \frac{1}{150} & 0 \\0 & 0 & 0 & 0 &  \frac{y}{150}\end{array}\right]\left[\begin{array}{c}z_1 \\ z_2 \\ z_3 \\ z_4 \\  z_5\end{array}\right]+ \left[\begin{array}{c}0 \\ 45 \\ 0 \\ 45 \\  0\end{array}\right]\\[0.2cm]
&=: C(y)z+c,
\end{align*}
where $y>0$ is a parameter affecting the travel time on road~$5$.

We note that with this choice of travel time functions Assumption \ref{traffic} is satisfied because we get 
$F(x,y)=[C(y)z(x)+c+C(y)x^i]_{i=1}^N$
where $C(y)\succ 0$ for all $y>0$. Consequently, by the properties of the Kronecker product,
\begin{equation}\nabla_x F(x,y)=[I_N+\mathbbm{1}_N \mathbbm{1}_N^\top] \otimes  C(y) \succ 0, \forall y>0,
\label{step_mon}
\end{equation}
since both $C(y)$ and $[I_N+\mathbbm{1}_N \mathbbm{1}_N^\top]$ are positive definite. 
Equation \eqref{step_mon} is a sufficient condition for $x\mapsto F(x,y)$ to be strongly monotone for every $y>0$.

 We consider a  game with $N=12$ agents and we assume that each agent $i$ has flow $\eta^i=12.5$ and wants to go from node $1$ to $4$. Moreover, we assume that only a fraction $q\in[0,1]$ of the agents  knows edge $5$ (i.e., $\mathcal{E}^i:=\mathcal{E}$  for $i\in[1,qN]$),
while the remaining agents don't (i.e., $\mathcal{E}^i:=\{1,2,3,4\}$ for $i\in[qN+1,N]$). The plot at the top of Figure \ref{tttN} shows the total travel time $s(y)$ at the Nash equilibrium as a function of the parameter $y$, for the $4$ different  scenarios in which the percentage $q$ of agents that  know road $5$ varies in $\{1,2/3,1/3,0\}$. The plot at the bottom shows the sensitivity of $s(y)$  computed according to \eqref{sensitivity_ttt}. First note that the \textit{Braess' paradox} is captured by the fact that for $q=1$ (i.e when all the agents know all the network) $\frac{\partial   s(y)}{\partial y}$ is negative, that is, augmenting the  cost of edge $5$ decreases the total travel time. Another interesting aspect that is illustrated by this picture is the \textit{informational Braess' paradox} \cite{acemoglu2016informational}, that is, the fact that reducing the information given to the agents (i.e. decreasing $q$) actually decreases the total travel time, with the limit result that if no agent knew edge $5$ (i.e. $q=0$) then they would all be better off.

\begin{figure}
\begin{center}
\includegraphics[width=0.48\textwidth]{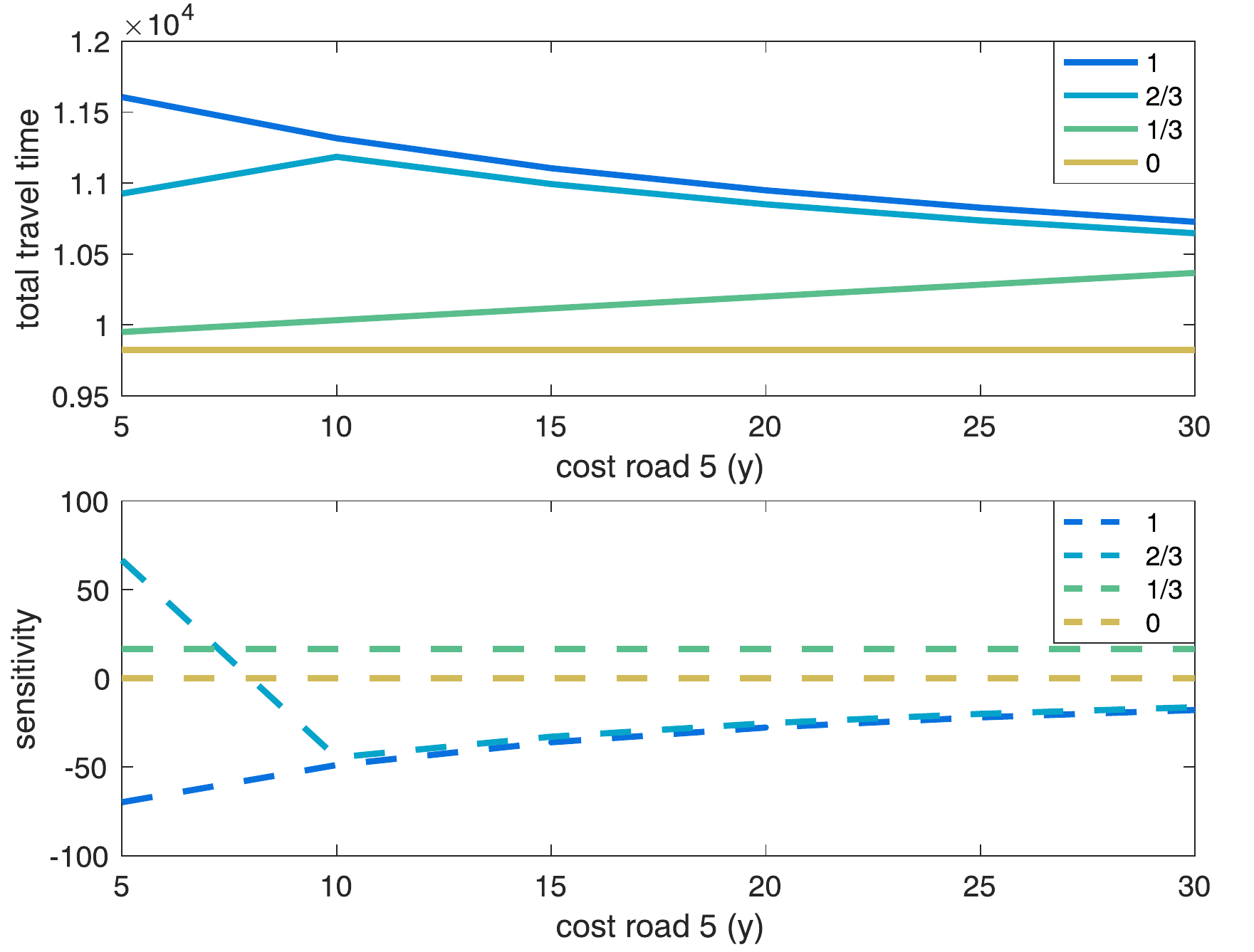}\\
 \end{center}
\caption{\small Total travel time  as a function of the  cost of road $5$ for a population of size $N=12$ and the road network in Fig. \ref{net}.}
\label{tttN}
\end{figure}

\section{ Conclusion}
By using a formula for the sensitivity of the Nash equilibrium in terms of primal variables only,
we significantly extended previous sensitivity results on quadratic network games and atomic 
routing games with  information constraints.  Moreover, our findings shed new light on the inter-centrality or key-player measure introduced in \cite{ballester2006s}.  In this paper we focused on small variations in the cost functions of the agents.
As future direction we aim at investigating the effect of macroscopic variations of other game primitives, as for example the network topology.

\bibliographystyle{IEEEtran}

\bibliography{mit.bib}

\begin{thebibliography}{10}
\providecommand{\url}[1]{#1}
\csname url@samestyle\endcsname
\providecommand{\newblock}{\relax}
\providecommand{\bibinfo}[2]{#2}
\providecommand{\BIBentrySTDinterwordspacing}{\spaceskip=0pt\relax}
\providecommand{\BIBentryALTinterwordstretchfactor}{4}
\providecommand{\BIBentryALTinterwordspacing}{\spaceskip=\fontdimen2\font plus
\BIBentryALTinterwordstretchfactor\fontdimen3\font minus
  \fontdimen4\font\relax}
\providecommand{\BIBforeignlanguage}[2]{{%
\expandafter\ifx\csname l@#1\endcsname\relax
\typeout{** WARNING: IEEEtran.bst: No hyphenation pattern has been}%
\typeout{** loaded for the language `#1'. Using the pattern for}%
\typeout{** the default language instead.}%
\else
\language=\csname l@#1\endcsname
\fi
#2}}
\providecommand{\BIBdecl}{\relax}
\BIBdecl

\bibitem{ghaderi2014opinion}
J.~Ghaderi and R.~Srikant, ``{Opinion dynamics in social networks with stubborn
  agents: Equilibrium and convergence rate},'' \emph{Automatica}, vol.~50,
  no.~12, pp. 3209--3215, 2014.

\bibitem{acemoglu2015networks}
D.~Acemoglu, A.~Ozdaglar, and A.~Tahbaz-Salehi, ``Networks, shocks, and
  systemic risk,'' in \emph{The Oxford Handbook of the Economics of Networks},
  Y.~Bramoull{\'e}, A.~Galeotti, and B.~Rogers, Eds.\hskip 1em plus 0.5em minus
  0.4em\relax Oxford University Press, 2015.

\bibitem{roughgarden2007routing}
T.~Roughgarden, ``Routing games,'' in \emph{Algorithmic game theory}, N.~Nisa,
  T.~Roughgarden, E.~Tardos, and V.~V. Vazirani, Eds.\hskip 1em plus 0.5em
  minus 0.4em\relax Cambridge University Press New York, 2007, vol.~18, pp.
  459--484.

\bibitem{ma:callaway:hiskens:13}
Z.~Ma, D.~S. Callaway, and I.~A. Hiskens, ``Decentralized charging control of
  large populations of plug-in electric vehicles,'' \emph{IEEE Trans. on
  Control Systems Technology}, vol.~21, no.~1, pp. 67--78, 2013.

\bibitem{chen2014autonomous}
H.~Chen, Y.~Li, R.~H.~Y. Louie, and B.~Vucetic, ``Autonomous demand side
  management based on energy consumption scheduling and instantaneous load
  billing: An aggregative game approach,'' \emph{IEEE Trans. on Smart Grid},
  vol.~5, no.~4, pp. 1744--1754, 2014.

\bibitem{friesz2016foundations}
T.~L. Friesz and D.~Bernstein, \emph{Foundations of network optimization and
  games}.\hskip 1em plus 0.5em minus 0.4em\relax Springer, 2016.

\bibitem{dafermos1988sensitivity}
S.~Dafermos, ``Sensitivity analysis in variational inequalities,''
  \emph{Mathematics of Operations Research}, vol.~13, no.~3, pp. 421--434,
  1988.

\bibitem{parise2017network}
F.~Parise, B.~Gentile, and J.~Lygeros, ``A distributed algorithm for average
  aggregative games with coupling constraints,'' \emph{arXiv preprint
  arXiv:1706.04634}, 2017.

\bibitem{paccagnan2016aggregative}
D.~Paccagnan, M.~Kamgarpour, and J.~Lygeros, ``On aggregative and mean field
  games with applications to electricity markets,'' in \emph{European Control
  Conference}, 2016.

\bibitem{gentile2017nash}
B.~Gentile, F.~Parise, D.~Paccagnan, M.~Kamgarpour, and J.~Lygeros, ``Nash and
  {W}ardrop equilibria in aggregative games with coupling constraints,''
  \emph{arXiv preprint arXiv:1702.08789}, 2017.

\bibitem{jackson2014games}
M.~O. Jackson and Y.~Zenou, ``Games on networks,'' in \emph{Handbook of game
  theory}, P.~Young and S.~Zamir, Eds., 2014, vol.~4.

\bibitem{ballester2006s}
C.~Ballester, A.~Calv{\'o}-Armengol, and Y.~Zenou, ``Who's who in networks.
  {W}anted: The key player,'' \emph{Econometrica}, vol.~74, no.~5, pp.
  1403--1417, 2006.

\bibitem{acemoglu2016informational}
D.~Acemoglu, A.~Makhdoumi, A.~Malekian, and A.~Ozdaglar, ``Informational
  {B}raess' paradox: The effect of information on traffic congestion,''
  \emph{arXiv preprint arXiv:1601.02039}, 2016.

\bibitem{qiu1989sensitivity}
Y.~Qiu and T.~L. Magnanti, ``Sensitivity analysis for variational inequalities
  defined on polyhedral sets,'' \emph{Mathematics of Operations Research},
  vol.~14, no.~3, pp. 410--432, 1989.

\bibitem{josefsson2007sensitivity}
M.~Josefsson and M.~Patriksson, ``Sensitivity analysis of separable traffic
  equilibrium equilibria with application to bilevel optimization in network
  design,'' \emph{Transportation Research Part B: Methodological}, vol.~41,
  no.~1, pp. 4--31, 2007.

\bibitem{lu2008sensitivity}
S.~Lu, ``Sensitivity of static traffic user equilibria with perturbations in
  arc cost function and travel demand,'' \emph{Transportation science},
  vol.~42, no.~1, pp. 105--123, 2008.

\bibitem{do2014sensitivity}
B.~Do~Chung, H.-J. Cho, T.~L. Friesz, H.~Huang, and T.~Yao, ``Sensitivity
  analysis of user equilibrium flows revisited,'' \emph{Networks and Spatial
  Economics}, vol.~14, no.~2, pp. 183--207, 2014.

\bibitem{dubey2006strategic}
P.~Dubey, O.~Haimanko, and A.~Zapechelnyuk, ``Strategic complements and
  substitutes, and potential games,'' \emph{Games and Economic Behavior},
  vol.~54, no.~1, pp. 77--94, 2006.

\bibitem{milgrom1994monotone}
P.~Milgrom and C.~Shannon, ``Monotone comparative statics,''
  \emph{Econometrica: Journal of the Econometric Society}, pp. 157--180, 1994.

\bibitem{jensen2010aggregative}
M.~K. Jensen, ``Aggregative games and best-reply potentials,'' \emph{Economic
  theory}, vol.~43, no.~1, pp. 45--66, 2010.

\bibitem{acemoglu2013aggregate}
D.~Acemoglu and M.~K. Jensen, ``Aggregate comparative statics,'' \emph{Games
  and Economic Behavior}, vol.~81, pp. 27--49, 2013.

\bibitem{bramoulle2014strategic}
Y.~Bramoull{\'e}, R.~Kranton, and M.~D'Amours, ``Strategic interaction and
  networks,'' \emph{The American Economic Review}, vol. 104, no.~3, pp.
  898--930, 2014.

\bibitem{scutari2010convex}
G.~Scutari, D.~P. Palomar, F.~Facchinei, and J.-S. Pang, ``Convex optimization,
  game theory, and variational inequality theory,'' \emph{IEEE Signal
  Processing Magazine}, vol.~27, no.~3, pp. 35--49, 2010.

\bibitem{facchinei2007finite}
F.~Facchinei and J.-S. Pang, \emph{{Finite-dimensional variational inequalities
  and complementarity problems}}, ser. Operations Research.\hskip 1em plus
  0.5em minus 0.4em\relax Springer, 2003.

\bibitem{zeng2001property}
L.~Zeng, ``A property of the {L}eontief inverse and its applications to
  comparative static analysis,'' \emph{Economic Systems Research}, vol.~13,
  no.~3, pp. 299--315, 2001.

\bibitem{friedkin1999social}
N.~E. Friedkin and E.~C. Johnsen, ``Social influence networks and opinion
  change,'' \emph{Advances in group processes}, vol.~16, no.~1, pp. 1--29,
  1999.

\end{thebibliography}

\end{document}